\documentclass[journal]{IEEEtran}
\ifCLASSINFOpdf
\else
\fi
\usepackage{graphicx} 
\usepackage{amsmath}
\usepackage{amssymb}
\usepackage{amsthm}

\newtheorem{mydef}{Definition}
\newtheorem{theorem}{Theorem}
\newtheorem{remark}{Remark} 
\newtheorem{lem}{Lemma}
\usepackage{multicol}
\newtheorem{assum}{Assumption}

\usepackage{multirow}
\usepackage{color}
\usepackage{multicol}
\usepackage[T1]{fontenc}
\usepackage{algorithm,algorithmic}
\usepackage{epstopdf}
\usepackage{cite}

\begin{document}
%
\title{Memory-Based Data-Driven MRAC Architecture Ensuring Parameter Convergence}
%
%
%

\author{Sayan Basu Roy,
        Shubhendu Bhasin, Indra Narayan Kar    

\thanks{Sayan Basu Roy, Shubhendu Bhasin and Indra Narayan Kar are with the Department of Electrical Engineering, Indian Institute of Technology Delhi, New Delhi, India e-mail:(sayanetce@gmail.com, sbhasin@ee.iitd.ac.in, ink@ee.iitd.ac.in).}}    

\maketitle

\begin{abstract}
 Convergence of controller parameters in standard model reference adaptive control (MRAC) requires the system states to be persistently exciting (PE), a restrictive condition to be verified online. A recent data-driven approach, concurrent learning, uses information-rich past data concurrently with the standard parameter update laws to guarantee parameter convergence without the need of the PE condition. This method guarantees exponential convergence of both the tracking and the controller parameter estimation errors to zero, whereas, the classical MRAC merely ensures asymptotic convergence of tracking error to zero. However, the method requires knowledge of the state derivative, at least at the time instances when the state values are stored in memory. The method further assumes knowledge of the control allocation matrix. This paper addresses these limitations by using a memory-based finite-time system identifier in conjunction with a data-driven approach, leading to convergence of both the tracking and the controller parameter estimation errors without the PE condition and knowledge of the system matrices and the state derivative. A Lyapunov based stability proof is included to justify the validity of the proposed data-driven approach. Simulation results demonstrate the efficacy of the suggested method.
\end{abstract}


%
\IEEEpeerreviewmaketitle

\section{Introduction}
%
%
%
%
\IEEEPARstart{T}{he} design objective of Model Reference Adaptive Control (MRAC) is to make the system imitate the response of a chosen reference model. Classical and many recent MRAC techniques that use merely instantaneous data for adaptation (see \cite{Nr, An, Tao, Wt, Cao} and references there in) require that the system states be persistently exciting (PE) to ensure the convergence of the parameter estimates to their true values \cite{Tao}. In \cite{Boyd}, Boyd and Sastry proved that the PE condition on the  regressor translates to the reference input having as many spectral lines as the number of unknown parameters, however, the condition is rather restrictive. Enforcing the PE condition through exogenous excitation of the input is not always realizable and it is often impractical to monitor online whether a signal will remain PE as the condition depends on the future values of the signal. Since parameter convergence under the PE condition is difficult to apply, various algorithms like e-modification, $\sigma$-modification etc. are proposed in literature to guarantee boundedness of the parameters  \cite{An},\cite{Kok}. \par
Similar to MRAC, other model-based control methods such as model-based reinforcement learning (MBRL) \cite{Abb, Mitr, Mar, Dix, Bhasin2}, and model-based predictive control (MPC)\cite{Fuku, Guay3, Giri7, Asw}, require the controller to be developed based on the estimates of the unknown parameters. Therefore, the stability of the closed-loop system and the performance of the control law in all these cases crucially rely on parameter convergence, which requires restrictive PE condition. \par
Recent works \cite{Lewis}-\cite{Giri2} on learning and data-driven control methods have shown promise in improving tracking performance as they use input-output data along the system trajectory which carries sufficient information about the unknown system and the controller parameters. Girish et al. \cite{Giri1, Giri2, Giri9, Giri8} proposed a novel approach, coined as concurrent learning (CL), where information-rich past data is stored and concurrently used along with gradient based parameter update laws. Although the parameter estimation error is not directly measurable, the intelligent introduction of the concurrent learning term computed from the past stored data is proportional to parameter estimation error.  A sufficient condition associated with the rank of a matrix formed out of stored data is required for parameter convergence. Unlike the PE condition, the rank condition is more realistic and guarantees exponential convergence of tracking and parameter estimation errors to zero. Moreover, CL-based techniques are employed in the context of adaptive optimal control in \cite{Dix}, \cite{Dix3}, \cite{Hes} and experimental success have been found in \cite{Giri10}, \cite{Dix2}.\par 
Although concurrent learning is a powerful online adaptive control method, it requires the state derivative information at the time points at which the state values are stored \cite{Giri1}. In many practical situations, the state derivative is not measurable. In \cite{Giri1}, an optimal fixed point smoothing technique is used to estimate derivative at past values using a forward and backward Kalman filter \cite{Gelb}. However, the estimation method requires storing several forward and backward data points in time leading to a high memory requirement. Further, the state derivative estimation error degrades the exponential convergence result to a weaker one of uniformly ultimately bounded (UUB) stability. Moreover, both the classical MRAC and the concurrent learning laws require knowledge of the $B$ matrix  (in the standard state space realization).\par
In \cite{Lav, Ma, Han, Som} the uncertainty in the control allocation matrix ($B$ in LTI framework) is dealt with in different ways, however, few results in literature tackle the general case of controlling dynamical systems when the knowledge of input matrix is absent. Some recent results \cite{Giri3},\cite{Bhasin} have attempted to address the limitations in the concurrent learning framework. In \cite{Giri3}, the authors designed a control law without requiring $B$, however, the knowledge of the $A$ matrix and the state derivative is required. The requirement of the state derivative is avoided in \cite{Bhasin} by the introduction of a dynamic state derivative estimator, which leads to a UUB result, while requiring knowledge of the $B$ matrix. \par
The contribution of this paper is to achieve the MRAC goal with parameter convergence, using only state and input data. In this work, the system matrices $A$ and $B$ are considered to be unknown and the state derivative information is also not available. Two memory stacks to store effective past data points are utilized to solve the data-driven MRAC problem, relaxing the assumption of knowledge of the state derivative and the input $B$ matrix. Using sufficient rank conditions on the matrices formed out of the stored data of memory stacks, finite-time identification of the system parameters and subsequent exponential convergence of tracking and controller parameter estimation errors is obtained. The identification method proposed in this work is inspired from \cite{Guay1}, \cite{Guay2}. The finite time identification of system parameters eliminates the need of the computationally burdensome purging algorithm \cite{Giri3}. By introducing an additional gain parameter in the parameter update law, this work further avoids the singular value maximisation algorithm \cite{Giri4}, used to continuously update the history stack for accelerated convergence. Moreover, it is proved that the aforementioned rank condition on the respective matrices merely demands the corresponding signals to be exciting for a finite time interval, which is less restrictive than the PE condition. Unlike PE, the rank condition is required for the matrices formed out of past stored signals and therefore can be verified online.
\section{Classical MRAC}
Consider a continuous-time LTI system given by
\begin{equation} \label{eq:plant}
\dot{x}(t)=Ax(t)+Bu(t)
\end{equation}
where $x(t)\in \mathbb{R}^n$ denotes the state and $u(t)\in \mathbb{R}^d$ denotes the control input to the system and $A\in \mathbb{R}^{n\times n}$, $B \in \mathbb{R}^{n\times d}$ are the system matrices. It is assumed that the pair $(A,B)$ is controllable and that $B$ has full column rank. \footnote{ Typically in physical systems $d\leq n$ and the above mentioned condition is satisfied.}\par
A reference model is chosen as follows to characterise the desired closed loop response of the system \eqref{eq:plant}.
\begin{equation} \label{eq:ref}
\dot{x}_m(t)=A_mx_m(t)+B_mr(t)
\end{equation}
where $A_m \in \mathbb{R}^{n\times n}$ is Hurwitz, $x_m(t)\in \mathbb{R}^n$ is the model state and $r(t)\in \mathbb{R}^d$ denotes a bounded, piecewise continuous reference input signal. An adaptive control law, comprising a linear feedback term and a linear feedforward term, is defined as \cite{An}
\begin{equation} \label{eq:control}
u(t)=K_x^T(t)x(t)+K_r^T(t)r(t)
\end{equation}
where $K_x(t)\in \mathbb{R}^{n \times d}$ and $K_r(t)\in \mathbb{R}^{d \times d}$. Substituting \eqref{eq:control} in \eqref{eq:plant} yields
\begin{equation} \label{eq:plant2}
\dot{x}=(A+BK_x^T)x(t)+BK_r^Tr(t)
\end{equation}
To facilitate the design objective of making system \eqref{eq:plant2} respond as the chosen reference model of \eqref{eq:ref}, the following matching condition is introduced \cite{An}, \cite{Tao}.
\begin{assum}
 There exists $K_x^*\in \mathbb{R}^{n \times d}$ and $K_r^*\in \mathbb{R}^{d \times d}$ such that
\begin{eqnarray} 
A+BK_x^{*T}=A_m \label{eq:match1}\\
BK_r^{*T}=B_m     \label{eq:match2}
\end{eqnarray}
\end{assum}
Using \eqref{eq:match1} and \eqref{eq:match2}, the closed-loop system in \eqref{eq:plant2} can be written as
\begin{equation} \label{eq:plant3}
\dot{x}=A_mx+B_mr+B\tilde{K}_x^Tx+B\tilde{K}_r^Tr
\end{equation}
where $\tilde{K}_x\triangleq K_x-K_x^*$ and $\tilde{K}_r\triangleq K_r-K_r^*$. The tracking error is defined as
\begin{equation} \label{eq:error}
e(t)\triangleq x(t)-x_m(t)
\end{equation}
Using \eqref{eq:ref}, \eqref{eq:plant3} and \eqref{eq:error}, the error dynamics is obtained as
\begin{equation} \label{eq:erdy}
\dot{e}=A_me+B\tilde{K}_x^Tx+B\tilde{K}_r^Tr
\end{equation}
The standard adaptive update laws for $K_x(t)$ and $K_r(t)$ are given as \cite{An}
\begin{eqnarray}
\dot{K}_x=-\Gamma_xxe^TPB \label{eq:gain1}\\
\dot{K}_r=-\Gamma_rre^TPB \label{eq:gain2}
\end{eqnarray}
where $\Gamma_x>0$ and $\Gamma_r>0$ denote positive definite learning rate matrices of appropriate dimension and $P\in \mathbb{R}^{n \times n}$ is a positive definite matrix satisfying the Lyapunov equation 
\begin{equation} \label{eq:Lypunov}
A_m^TP+PA_m+Q=0
\end{equation} 
for any given positive definite $Q\in \mathbb{R}^{n\times n}$.\\
The equations \eqref{eq:control} and \eqref{eq:gain1}-\eqref{eq:gain2} are the classical MRAC laws, which guarantee the tracking error $e(t)\rightarrow 0$ as $t\rightarrow \infty$. However, convergence of the controller parameters $K_x(t)$ and $K_r(t)$ to their true values ($K_x^*$ and $K_r^*$ respectively) is only guaranteed if a restrictive PE condition is satisfied \cite{An}, \cite{Wt}, \cite{Sun}. The persistence of excitation of a vector signal $x(t)$  is defined below \cite{Tao}.\\
\begin{mydef}
 A bounded vector signal $x(t)$ is exciting over an interval $[t,t+T]$ , $T>0$ and $t\geq t_0$ if $\exists$ $\alpha>0$ such that the following condition holds:
\begin{equation} 
\int_{t}^{t+T}x(\tau)x^T(\tau)d\tau\geq \alpha I \nonumber
\end{equation}
 where $I$ denotes an identity matrix.
\end{mydef}
\begin{mydef}
  A bounded vector signal $x(t)$ is persistently exciting (PE) if $\forall t\geq t_0$, $\exists$ $T>0$ and $\alpha >0$ such that:
\begin{equation}
\int_{t}^{t+T}x(\tau)x^T(\tau)d\tau\geq \alpha I \nonumber
\end{equation}
\end{mydef}
According to \cite{Boyd}, for the system states of \eqref{eq:plant} to be PE, the reference signal $r(t)$ must be sufficiently rich i.e. it must contain $n$ distinct frequencies if there are $2n$ unknown parameters. The PE condition is restrictive and difficult to verify online as it relies on the future behaviour of the dynamical systems.
\section{Memory-Based Data-Driven MRAC}
This section proposes a memory-based data-driven architecture for model reference adaptive control and identification of unknown LTI systems with controller parameter convergence. This architecture builds on the concurrent learning technique \cite{Giri2}, \cite{Giri1}, \cite{Giri5}, \cite{Giri6}, \cite{Buss} which utilizes past recorded data concurrently with the current data for adaptation. The concurrent learning method fruitfully utilises a memory stack for storing the state $x(t)$ and the reference signal $r(t)$ at different time points. A full rank condition on the matrices formed out of the memory stack guarantees exponential convergence of both the tracking and the controller parameter estimation errors. However, both classical MRAC and the concurrent learning frameworks assume knowledge of the input $B$ matrix. Concurrent learning, additionally, requires the state derivative information at least at the time instances when the state information is stacked in memory. Although there have been efforts to overcome these pitfalls, the solutions given by \cite{Giri3}, \cite{Bhasin} are partial. \par 
Motivated by the aforementioned limitations, this work achieves the MRAC goal only from input and state data. The system matrices ($A$ and $B$) as well as the state derivative information is considered to be unknown in this framework. The intelligent exploitation of multiple memory stacks results in finite-time identification of system parameters and subsequent exponential convergence of tracking and controller parameter estimation errors to zero. 
\subsection{Finite Time Identification of System Parameters}
The plant dynamics in (1) can be linearly parametrized as
\begin{equation} \label{eq:lp}
\dot{x}=Y(x,u)\theta
\end{equation}
where $Y(x,u)\in \mathbb{R}^{n\times n(n+d)}$ is the regressor matrix and $\theta\in \mathbb{R}^{n(n+d)}$ is a column vector of all the elements of $A$ and $B$ defined as
\[\theta\triangleq\left[\begin{array}{c}
vec(A^T)\\
vec(B^T)\end{array}\right]\]
where $vec(Z)\in \mathbb{R}^{ab}$ denotes the vectorization of a matrix $Z\in \mathbb{R}^{a\times b}$, obtained by stacking the columns of the matrix $Z$. The set of equations required for system parameter identification are described as \cite{Guay1}
\begin{align}
&\dot{\hat{x}}=Y\hat{\theta}+k_m\tilde{x}+m\dot{\hat{\theta}} \label{eq:xhat}\\
&\dot{\hat{\theta}}=k_\theta m^T(\tilde{x}-\gamma) \label{eq:thetahat}\\
&\dot{m}=Y-k_mm, \quad m(t_0)=0 \label{eq:m}
\end{align}
where $m(t)\in \mathbb{R}^{n\times n(n+d)}$ and $k_m$, $k_\theta$  are positive scalar gains and $\tilde{x}(t)$ and $\gamma(t)\in \mathbb{R}^n$ are defined as
\begin{align}
&\tilde{x}\triangleq x-\hat{x} \label{eq:xtilde}\\
&\gamma\triangleq \tilde{x}-m\tilde{\theta} \label{eq:gamma}
\end{align}
where $\tilde{\theta}(t)$ is given by
\begin{equation} \label{eq:thetatilde}
\tilde{\theta}\triangleq \theta-\hat{\theta}
\end{equation}
Differentiating \eqref{eq:gamma} and using \eqref{eq:xhat}-\eqref{eq:xtilde}, the following expression is obtained
\begin{equation} \label{eq:gammadot}
\dot{\gamma}=-k_m\gamma
\end{equation}
with $\gamma(t_0)=\tilde{x}(t_0)$.\\
Although $\gamma(t)$ in \eqref{eq:gamma} is unmeasurable, the use of the initial condition $m(t_0)=0$ in \eqref{eq:m} ensures that $\gamma(t)$ is available online by solving \eqref{eq:gammadot} with known initial condition.\par
An auxiliary variable $g(t) \in \mathbb{R}^n$ is defined as
\begin{equation} \label{eq:g}
g\triangleq m\hat{\theta}+\tilde{x}-\gamma
\end{equation}
Substituting $\gamma$ from \eqref{eq:gamma} in \eqref{eq:g} leads to
\begin{equation} \label{eq:gm}
g(t)=m(t)\theta, \quad\forall t\geq t_0
\end{equation}
The strategic introduction of variables $m(t)$ and $\gamma(t)$ transforms the system in \eqref{eq:lp} to a standard linear regression form of \eqref{eq:gm}, where both $m(t)$ and $g(t)$ are known signals, obviating the need of $\dot{x}(t)$ information.\par
Consider a memory stack $W=\{(m_j,g_j)\}_{j=1}^{p}$ updated online with the signals $m(t)$ and $g(t)$ respectively, where each $m_i$ and $g_i$ are collected and stored in $W$ at $t=t_{w_i}$ with $t_0\geq t_{w_1}>t_{w_2}>....t_{w_p}=t_c$ .\footnote{The data storing mechanism in this work is analogous to that of \cite{Bhasin}, \cite{Giri4} and \cite{Giri6}, hence, the details of this mechanism are omitted here.} Here, $p$ is the memory stack length, which satisfies the condition $p\geq n+d$, where $n$ is the dimension of the state $x(t)$ and $d$ is the dimension of the input $u(t)$.
\begin{assum} \label{as:M}
The matrix $M\triangleq[m_1^T, m_2^T,...m_p^T]^T\in \mathbb{R}^{np\times n(n+d)}$ is full rank i.e. $rank(M)=n(n+d)$.
\end{assum}
This rank condition is analogous to those of CL-based frameworks \cite{Giri1}, \cite{Giri2} of adaptive control. The difference here is that it is stated in terms of a newly introduced variable $m(t)$ required for system identification purpose. Unlike the dependency on future behaviour of signals in PE condition, this rank condition relies on past data. Therefore, this condition can be easily verified online making it more practical as compared to the PE condition. The following Lemma establishes only exciting condition (not PE) of relevant signals as a necessary condition for the Assumption \ref{as:M} to hold.
\begin{lem}
The necessary condition for the matrix $M$ to be full rank is the state $x(t)$ and the input $u(t)$ be exciting over the interval $t\in [t_0,t_c]$ as per Definition 1.\footnote{For proof see Appendix}
\end{lem} 

\begin{theorem}
The system parameter identification error $\|\tilde{\theta}(t)\|$ is non-increasing during the interval $[t_0,t_c]$ using the update laws \eqref{eq:xhat}-\eqref{eq:m}. Provided Assumption \ref{as:M} holds, the stored data in $W$ can be used to achieve the finite time identification of system parameters ($\theta$) at $t=t_c$.\end{theorem}
\begin{proof}
 Consider a Lyapunov candidate as
\begin{equation} \label{eq:Lyaptheta}
V_{\theta}=\frac{1}{2}\tilde{\theta}^T\tilde{\theta}
\end{equation} 
Differentiating \eqref{eq:Lyaptheta} along \eqref{eq:thetahat} and using \eqref{eq:gamma} yields
\begin{equation}
\dot{V}_{\theta}=-k_\theta\tilde{\theta}^Tm^Tm\tilde{\theta}\leq 0
\end{equation}
which implies $\| \tilde{\theta}(t)\|\leq\| \tilde{\theta}(t_0)\|$, $\forall  t\geq t_0$.\\ 
For finite time identification of system parameters, define a matrix $G\triangleq[g_1^T,g_2^T,...g_p^T]^T\in \mathbb{R}^{np\times 1}$.\\
Since \eqref{eq:gm} is valid for $\forall t\geq t_0$, the following equation is satisfied.
\begin{equation}
M\theta=G
\end{equation}
The Assumption \ref{as:M} implies $M^TM$ is an invertible matrix and hence the system parameter $\theta$ can be found from the following least-square like expression.
\begin{equation} \label{eq:intr7}
\hat{\theta}_{FT}(t)=(M^TM)^{-1}M^TG=\theta, \quad t\geq t_c
\end{equation} 
Using \eqref{eq:intr7},at $t=t_c$, finite time convergence of $\theta$ is obtained.\\ 
The identification equations \eqref{eq:xhat}-\eqref{eq:m} and \eqref{eq:g} are merely required for $t\in [t_0,t_c]$. At $t=t_c$, using \eqref{eq:intr7} finite time identification of system parameters is achieved using the memory stack $W$. The finite-time (FT) identifier is given by
\begin{align} \label{eq:FT}
\hat{\theta}_{FT}(t)=
  \begin{cases}
  \hat{\theta}(t) & \text{if}\quad t<t_c\\
  \theta & \text{if} \quad t\geq t_c
  \end{cases}
\end{align}
\end{proof}
The FT identifier method requires instantaneous inversion of $M^TM\in \mathbb{R}^{n(n+d)\times n(n+d)}$ at $t=t_c$, which may be impractical for large dimensional systems. However, the following Lemma shows that the matrix $M^TM$ is significantly sparse and sparsity increases with the state dimension $n$ implying that sophisticated techniques can be applied for fast computation of the inverse.
\begin{lem}
The fraction of non-zero elements in $M^TM$ is $\frac{1}{n}$, where $n$ is the state dimension.
\end{lem}


\subsection{Tracking Error and Controller Parameter Convergence}
The control law in \eqref{eq:control} can be linearly parametrized as 
\begin{equation} \label{eq:lp2}
u=z(x,r)\phi(t)
\end{equation}
where $z(x,r)\in\mathbb{R}^{d\times d(n+d)}$ and $\phi\in \mathbb{R}^{d(n+d)}$ is a column vector consisting of all the elements of $K_x(t)$ and $K_r(t)$ defined as
\[\phi\triangleq\left[\begin{array}{c}
vec(K_x)\\
vec(K_r)\end{array}\right]\]
Using \eqref{eq:lp2}, the error dynamics in \eqref{eq:erdy} can be written as
\begin{equation} \label{eq:erdy2}
\dot{e}=A_me+Bz\tilde{\phi}
\end{equation}
where $\tilde{\phi}=\phi-\phi^*$ and
\[\phi^*\triangleq\left[\begin{array}{c}
vec(K_x^*)\\
vec(K_r^*)\end{array}\right]\]
For the convergence of tracking and controller parameter estimation errors, another memory stack $H=\{x_j, r_j\}_{j=1}^q$ is populated online by state and input signals, respectively where each $x_j$ and $r_j$ is stored at $t=t_{h_j}$ with $t_0\geq t_{h_1}>t_{h_2}>..>t_{z_q}=t_s$ . Each pair $(x_j, r_j)$ is referred to as a data point and $q$ is the length of the stack satisfying $q\geq n\geq d$.
\begin{assum} \label{as:Z}
The matrix $Z=[z_1^T, z_2^T,...,z_q^T]^T$ is full rank i.e. $rank(Z)=d(n+d)$, where $z_j=z(x_j, r_j)$ for $j=1(1)q$. 
\end{assum}
This assumption is analogous to Assumption \ref{as:M}. Here it is stated in terms of $z(t)$, which is relevant to controller parameter convergence. The following Lemma provides a necessary condition for the Assumption \ref{as:Z} to hold.
\begin{lem}
The necessary condition for the matrix $Z$ to be full rank is the state $x(t)$ and the reference input $r(t)$ be exciting over the interval $t\in [t_0,t_s]$ as per Definition 1. 
\end{lem}
The controller parameter $\phi(t)$ is updated as
\begin{align} \label{eq:phi}
\dot{\phi}=
  \begin{cases}
  proj(-\Gamma_{\phi}z^T\hat{B}^TPe) \quad \quad \text{for}\quad t_0\leq t\leq t_c\\
  -\Gamma_{\phi}z^T\hat{B}^TPe \quad \quad \quad \text{for} \quad t_c < t < t_m\\
  -\Gamma_{\phi}\left(z^T\hat{B}^TPe+k_{\phi}\sum_{j=1}^{q}z_j^T\epsilon_{K_j}\right) \quad \text{for} \quad t\geq t_m
  \end{cases}
\end{align}
where $\Gamma_\phi\in\mathbb{R}^{d(n+d)\times d(n+d)}$ is a positive definite learning rate matrix, $proj$ denotes projection operator \cite{Kr} which ensures parameter boundedness within a convex region in the parameter space and $t_m=max(t_c, t_s)$. Further, $k_{\phi}$ is a scalar gain introduced to alter the rate of convergence and  $\hat{B}(t)$ is extracted from $\hat{\theta}_{FT}(t)$ of \eqref{eq:FT}. The error variable $\epsilon_{K_j}(t)$ in \eqref{eq:phi} is defined as
\begin{equation} \label{eq:ek}
\epsilon_{K_j}(t)=\epsilon_{K_{x_j}}(t)+\epsilon_{K_{r_j}}(t)
\end{equation}
with the following two expressions.
\begin{align}
&\epsilon_{K_{x_j}}(t)\triangleq (\hat{B}^T\hat{B})^{-1}\hat{B}^T(\hat{\dot{x}}_j-A_mx_j-B_mr_j
-\hat{B} \epsilon_{K_{r_j}}(t)) \label{eq:ex}\\
&\epsilon_{K_{r_j}}(t)\triangleq K_r^T(t)r_j-(\hat{B}^T\hat{B})^{-1}\hat{B}^TB_mr_j \label{eq:er}
\end{align} 
where $\hat{\dot{x}}_j=Y_j\hat{\theta}_{FT}$. \\
\begin{theorem}
For the system \eqref{eq:plant}, the control law in \eqref{eq:lp2} and the update laws in \eqref{eq:phi} along with the finite-time system identifier $\hat{\theta}_{FT}(t)$ \eqref{eq:FT} ensure boundedness of the tracking and the controller parameter estimation errors for $t\in [t_0, t_m)$ and the global exponential convergence of those errors to zero is guaranteed for $t\geq t_m$, provided the Assumption \ref{as:Z} is satisfied.
\end{theorem}
\begin{proof}
Consider the following Lyapunov candidate
\begin{equation} \label{eq:Lyphi}
V_{\xi}=\frac{1}{2}\xi^T\Lambda \xi
\end{equation}
where $\xi(t)\triangleq [e^T(t), \phi^T(t)]^T$ and 
\[\Lambda\triangleq\left[\begin{array}{cc}
P & 0_{n\times d(n+d)}\\
0_{d(n+d)\times n} & \Gamma_{\phi}^{-1}\end{array}\right]\]
Taking time derivative of \eqref{eq:Lyphi} along the trajectories of \eqref{eq:erdy2} and \eqref{eq:phi} the following expression is obtained during $t\in [t_0,t_c]$
\begin{equation}
\dot{V}_{\xi}\leq -\frac{1}{2}e^TQe+e^TPBZ\tilde{\phi}-\tilde{\phi}^Tz^T\hat{B}^TPe
\end{equation} 
The inequality occurs due to the use of projection operator (For details see \cite{Lav2}). The above inequality can be further modified as
\begin{equation}
\dot{V}_{\xi}\leq -\frac{1}{2}e^TQe+e^TP\tilde{B}Z\tilde{\phi}
\end{equation}
where $\tilde{B}\triangleq B-\hat{B}$. Using the fact $Z\tilde{\phi}=\tilde{K_x}^Tx+\tilde{K_r}^Tr$, yields
\begin{equation} \label{eq:intr}
\dot{V}_{\xi}\leq -\frac{1}{2}\beta_1\|e\|^2+\beta_2\|e\|
\end{equation}
where
\begin{align}
&\beta_1=\lambda_{min}(Q)-2\|P\|\|\tilde{B}\|\|\tilde{K}_x\| \label{eq:beta1}\\
&\beta_2=\|P\|\|\tilde{B}\|(\|\tilde{K}_x\|\|x_m\|+\|\tilde{K}_r\|\|r\|)\label{eq:beta2}
\end{align}
$\lambda_{min}(.)$, in \eqref{eq:beta1}, denotes the minimum eigen value of the corresponding argument matrix. In \eqref{eq:beta2}, $\|\tilde{K}_x(t)\|$ and $\|\tilde{K}_r(t)\|$ are bounded by the use of projection operator \cite{Kr} in \eqref{eq:phi} and $\|\tilde{B}(t)\|\in\mathcal{L}_\infty$ from Theorem 1. As $r(t)\in\mathcal{L}_\infty$ and $A_m$ is Hurwitz by definition, $x_m(t)\in\mathcal{L}_\infty$, implying $\beta_1(t)\in\mathcal{L}_\infty$ and $\beta_2(t)\in\mathcal{L}_\infty$ with $\beta_2>0$. However, the sign of $\beta_1$ is uncertain during $t\in[t_0, t_c]$. It can be inferred that once $\beta_1(t)$ becomes greater than zero, it will remain greater than zero as $\|\tilde{\theta}(t)\|$ is non-increasing in the interval $[t_0,t_c]$ as per Theorem 1. Moreover, $\beta_1=\lambda_{min}(Q)>0$ at $t=t_c$ as $\|\tilde{\theta}(t)\|=0$ at $t=t_c$. Three cases are possible depending on the $\tilde{\theta}(t)$ dynamics.\\
\textbf{case 1}: $\beta_1(t)\leq 0$, $\forall t\in[t_0, t_c)$\\
\textbf{case 2:} $\beta_1(t)>0$, $\forall t\in[t_0, t_c]$\\
\textbf{case 3:} $\beta_1(t)\leq 0$, $\forall t\in[t_0, t_f]$ and $\beta_1(t)>0$, $\forall t\in(t_f,t_c]$ \\
The three cases are analysed separately as follows.\par
\textbf{case 1:} As $\beta_1(t)\leq 0$, it is hard to comment on the exact bound of the tracking error. However, finite tracking error can be claimed from \eqref{eq:plant2}, which can be expressed as $\dot{x}=\bar{A}(t)x(t)+\bar{g}(t)$, where $\bar{A}(t)=A+BK_x^T(t)$ and $\bar{g}(t)=BK_r^T(t))r(t)$. As equation \eqref{eq:plant2} is a linear equation in $x(t)$ with $\|\bar{A}(t)\|$ and $\|\bar{g}(t)\|$ are bounded i.e. $\bar{A}(t)\in\mathcal{L}_\infty$ and $\bar{g}(t)\in\mathcal{L}_\infty$ in the finite time interval $[t_0, t_c]$, using Global Existence and Uniqueness theorem \cite{Khalil} it can be argued that $x(t)$ cannot have a finite escape time. Therefore $x(t)\in\mathcal{L}_\infty$ in finite time ($t\leq t_c$) if $x(t_0)$ is finite, leading to the tracking error $e(t)\in\mathcal{L}_\infty$ as $x_m(t)\in\mathcal{L}_\infty$, implying $\xi(t)\in\mathcal{L}_\infty$.\par 
\textbf{case 2:} As $\beta_1(t)>0$, \eqref{eq:intr} can be further modified to
\begin{equation} \label{eq:intr2}
\dot{V}_{\xi}\leq -\frac{1}{2} \beta_{11}\|e\|^2+\frac{\beta_{2,max}^2}{4\beta_{12}}
\end{equation}
where $\beta_{11}+\beta_{12}=\beta_{1,mi n}>0$ with $\beta_{11}>0$, $\beta_{12}>0$ and $\beta_{2,max}$ can be found by upper bounding every time-varying term in \eqref{eq:beta2}. Due to the use of $proj$ in \eqref{eq:phi} during $t\in[t_0, t_c]$, $V_{\xi}$ in \eqref{eq:Lyphi} can be upper bounded as
\begin{equation} \label{eq:intr3}
V_{\xi}\leq\frac{1}{2}\lambda_{max}(P)\|e\|^2+D
\end{equation}
where $D=\frac{1}{2}\lambda_{min}(\Gamma_{\phi})\|\phi\|_{max}^2$.\footnote{$\|\phi\|_{max}^2$ is defined by the convex region in the projection operator.} Further upper bounding \eqref{eq:intr2} using \eqref{eq:intr3} yields
\begin{equation} 
\dot{V}_{\xi}\leq -\eta_1V_{\xi}+\eta_2
\end{equation}
where $\eta_1=\frac{\beta_{11}}{\lambda_{max}(P)}$ and $\eta_2=\frac{\beta_{2,max}^2}{4\beta_{12}}+\frac{\beta_{11}D}{\lambda_{max}(P)}$. Using comparison Lemma \cite{Khalil}, the above differential inequality results in the following UUB condition.
\begin{equation} \label{eq:uub}
V_{\xi}(t)\leq \left(V_{\xi}(t_0)-\frac{\eta_2}{\eta_1}\right)exp    \left(-\eta_1(t-t_0)\right)+\frac{\eta_2}{\eta_1}, \forall t\in[t_0,t_c]
\end{equation}
The inequality in \eqref{eq:uub} implies $\xi(t)\in\mathcal{L}_\infty$ during the same interval via Theorem 4.18 of \cite{Khalil}.\\
The analysis done in case 1 holds independent of the sign of $\beta_1(t)$. However, in case of $\beta_1(t)>0$, it is possible to get an exact expression of bound for the Lyapunov function $V_\xi(t)$ as shown in \eqref{eq:uub}.\par 
\textbf{case 3:} During $t\in[t_0, t_f+\epsilon)$, following the arguments similar to case 1, it can be established that $\xi(t)\in\mathcal{L}_\infty$, where $\epsilon>0$ is  infinitesimally small. Further, during $t\in[t_f+\epsilon, t_c]$, following the analysis similar to case 2, the following bound can be derived. 
\begin{equation} \label{eq:uub2}
V_{\xi}(t)\leq \left(V_{\xi}(t_\beta)-\frac{\eta_2}{\eta_1}\right)exp    \left(-\eta_1(t-t_\beta)\right)+\frac{\eta_2}{\eta_1}, \forall t\in[t_\beta,t_c]
\end{equation}
with $t_\beta=t_f+\epsilon$, implying $\xi(t)\in\mathcal{L}_\infty$ in the same interval.\\
In the interval $t\in (t_c,t_m)$ using $\tilde{B}=0$ from Theorem 1 in the time derivative of \eqref{eq:Lyphi} along \eqref{eq:erdy2} and \eqref{eq:phi}
\begin{equation}
\dot{V}_{\xi}=-\frac{1}{2}e^TQe\leq 0
\end{equation}
Thus $V_{\xi}(t)$ is non-increasing in this interval, implying $V_{\xi}(t_m)\leq V_{\xi}(t_c)$. Again, using Theorem 4.18 of \cite{Khalil} it can be inferred that $\xi(t)\in\mathcal{L}_\infty$ during $t\in(t_c, t_m)$ as $V_\xi(t)\in\mathcal{L}_\infty$.\\
For $t\geq t_m$, using $\hat{\theta}_{FT}=\theta$ from Theorem 1, the time derivative of \eqref{eq:Lyphi} along \eqref{eq:erdy2} and \eqref{eq:phi} can be expressed as
\begin{align} \label{eq:intr4}
&\dot{V}_{\xi}=-\frac{1}{2}e^TQe+e^TPBz\tilde{\phi}\nonumber\\
&-\tilde{\phi}^Tz^TB^TPe-\tilde{\phi}^T\left(k_{\phi}\sum_{j=1}^qz_j^Tz_j\right)\tilde{\phi}
\end{align}
Using \eqref{eq:match2} and \eqref{eq:er}, $\epsilon^T_{K_{r_j}}(t)$ can be expressed as
\begin{equation} \label{eq:e1}
\epsilon_{K_{r_j}}(t)=\tilde{K}_r^T(t)r_j
\end{equation}
and using \eqref{eq:plant3}, \eqref{eq:ex} and \eqref{eq:e} it can be shown that
\begin{equation} \label{eq:e2}
\epsilon_{K_{x_j}}(t)=\tilde{K}_x^T(t)x_j
\end{equation}
Further \eqref{eq:ek}, \eqref{eq:e1} and \eqref{eq:e2} lead to 
\begin{equation} \label{eq:e}
\epsilon_{K_j}(t)=z_j\tilde{\phi}(t)
\end{equation}
The expression \eqref{eq:e} is used to derive \eqref{eq:intr4}, which can be further upper bounded as
\begin{equation} \label{eq:intr5}
\dot{V}_{\xi}\leq -\frac{1}{2}\lambda_{min}(Q)\|e\|^2-k_{\phi}\lambda_{min}(\Omega_{z})\|\tilde{\phi}\|^2
\end{equation}
where 
\begin{equation}
\Omega_z=\sum_{j=1}^qz_j^Tz_j=Z^TZ
\end{equation}
Hence, based on Assumption \ref{as:Z}, $\Omega_z>0$ which implies $\dot{V}_{\xi}\leq 0$. Further from \eqref{eq:intr5} the following bound can be obtained
\begin{equation} \label{eq:intr6}
\dot{V}_{\xi}\leq -\beta V_{\xi}
\end{equation}
where $\beta$ is given by
\begin{equation}
\beta=\frac{min(\lambda_{min}(Q),2k_{\phi}\lambda_{min}(\Omega_z))}{max(\lambda_{max}(P),\lambda_{max}(\Gamma_{\phi}^{-1}))} \nonumber
\end{equation}
The differential inequality in \eqref{eq:intr6} leads to the subsequent exponentially convergent bound
\begin{equation}
V_{\xi}(t)\leq V_{\xi}(t_m)e^{-\beta(t-t_m)}, \forall t\in[t_m, \infty)
\end{equation}
implying $\xi(t)\to 0$ exponentially fast as $t\to \infty$. Further, the Lyapunov function in \eqref{eq:Lyphi} is radially unbounded and no restriction is imposed on $V_\xi(t_m)$, implying global exponential stability (GES).
\end{proof}
\begin{remark}
As described by \eqref{eq:phi}, the update law of $\phi(t)$ follows time-dependent switching with at most 2 (finite) switching instances ($t_c$ and $t_m$). Therefore, the boundedness of $\xi(t)$ during $t\in[t_0,t_m)$ and exponential convergence for $t\geq t_m$ suffices the analysis.
\end{remark}
\begin{remark}
To improve the rate of convergence, the memory stack is continuously updated in \cite{Giri4} using an algorithm to maximize the minimum singular value of a matrix analogous to $Z$. To avoid computational burden associated with the continuous stack update, the proposed algorithm updates the stack until the sufficient rank condition is satisfied. The speed of convergence is controlled by appropriately choosing $Q$, $k_{\phi}$ and $\Gamma_{\phi}$.  Moreover, since $\theta$ is obtained in finite time using \eqref{eq:xhat}-\eqref{eq:gm} and the memory stack $W$, the need for computationally involved purging algorithm \cite{Giri3} is obviated.
\end{remark}
\begin{remark}
The proposed memory-based data-driven technique for parameter convergence is similar to classical integral control in the following sense. It is well-known that integral control, which captures the effect of entire past of the relevant signal, reduces steady state error. The proposed approach also stores information-rich past data (although not the entire past) in the memory stack using a non-linear sampling technique and utilizes the stacked data in the parameter update law, leading to exponential convergence of parameter estimation error to zero. Future research can be carried out to investigate the relation between the proposed method and the classical integral control. 
\end{remark}
\section{Simulation Results}
To demonstrate the effectiveness of the proposed data driven technique, a second order linear plant is considered. 
\[A=\left[\begin{array}{cc}
0 & 1 \\
5 & 2 \end{array} \right]
B=\left[\begin{array}{c}
0 \\
2 \end{array} \right] \] 
The reference model matrices are considered as
\[A_m=\left[\begin{array}{cc}
0 & 1 \\
-8 & -10 \end{array} \right]
B_m=\left[\begin{array}{c}
0 \\
1 \end{array} \right] \] 
Note that $A_m$ is a Hurwitz matrix but $A$ is not. Using the matching conditions \eqref{eq:match1} and \eqref{eq:match2} $K_x^*$ is equal to $[-6.5, -6]^T$ and $K_r^*$ is $0.5$. The reference signal $r(t)$ is chosen  as
\begin{equation}
r(t)=20e^{-t/2} \nonumber
\end{equation} 
which is a non-PE signal. The matrix $Q$ of Lyapunov equation (12) is selected as
\[Q=\left[\begin{array}{cc}
5 & 0\\
0 & 5\end{array}\right]\]
The adaptation gains $\Gamma_x$ and $\Gamma_r$ are chosen as $\Gamma_x=I_2$ and $\Gamma_r=I_1$. The gain parameters are chosen as $k_{\theta}=80$, $k_m=10$ and $k_{\phi}=40$. \par
The plot of the error dynamics of system parameters is shown in Fig 1. The tracking error plot is shown in Fig 2, depicting the convergence of errors to zero within approximately 4.1 seconds. Fig 3. shows the evolution of estimation error in controller parameters. At $t=t_m$ the transition of parameter update law from one rule to another leads to non-differentiability at that time point. After $t=t_m$, the error dynamics converge to zero exponentially.
\begin{figure}[!t]
   \includegraphics[width=3 in, height=2 in]{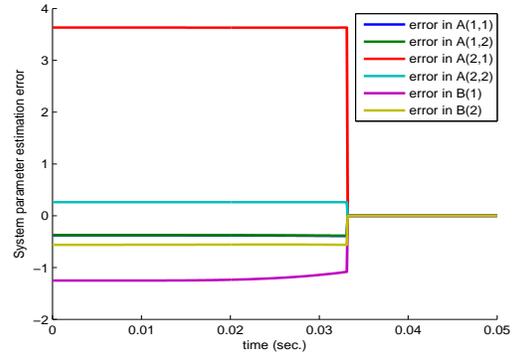}
   \caption{System parameter estimation error $\tilde{\theta}(t)$}
\end{figure}
\begin{figure}[!t]
   \includegraphics[width=3 in, height=2 in]{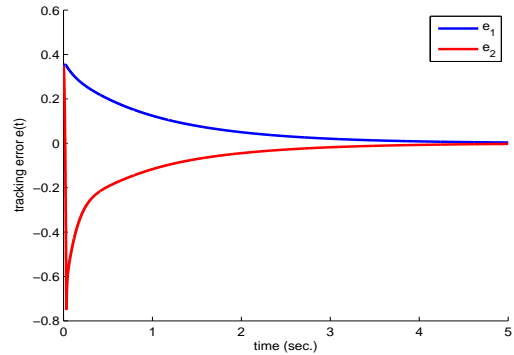}
   \caption{tracking error e(t)}
\end{figure}
\begin{figure}[!t]
   \includegraphics[width=3 in, height=2 in]{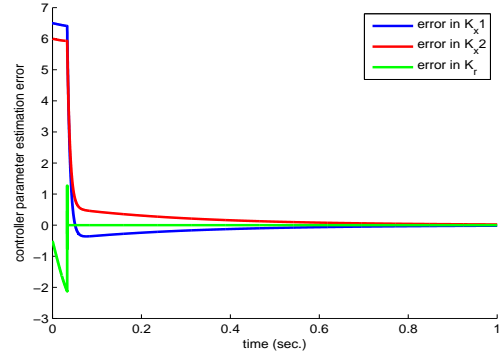}
   \caption{Controller parameter estimation error $\tilde{\phi}(t)$}
\end{figure}

\section{Conclusion}
A memory-based data-driven approach is proposed to solve the MRAC problem for unknown LTI systems, using only input and state data. Past data along the system trajectory is stored and used strategically to guarantee finite-time identification of system parameters ($\theta$), convergence of controller parameters ($\phi$) and tracking error. Unlike the restrictive PE condition in classical adaptive control, only a rank condition on the recorded matrices is required to ensure parameter convergence. Further, the work relaxes two vital assumptions of CL-based frameworks of adaptive control-knowledge of the state derivative and the input $B$ matrix, without altering the exponential convergence result.


%

\appendix
\begin{lem}
A memory stack $X$ is populated with the values of a continuous signal $x(t)\in R^{n \times m}$ at different time points as described below\\
$X=[x_1,x_2,....x_p]$, where $mp\geq n$ and $x_i=x(t_i)\in R^{n\times m}$, $i=1(1)p$ with $t_1>t_2>...t_p=t_e$. \\
If the matrix $X$ is full rank, i.e. $rank(X)=n$, the signal $x(t)$ is exciting over the interval $[t_0,t_e]$ as per Definition 1.
\end{lem}
\begin{proof}
Define $C\triangleq\int_0^{t_e} x(\tau)x^T(\tau)d\tau$. Taking limit as $\epsilon\to 0^+$ and using the corollary of fundamental theorem of calculus $\lim_{\epsilon\to 0^+}\frac{1}{2\epsilon}\int_{a-\epsilon}^{a+\epsilon}f(y)dy=f(a)$ for any continuous function $f$, the following expression can be obtained
\begin{align}
&\lim_{\epsilon\to 0^+}\frac{1}{2\epsilon}C=\lim_{\epsilon\to 0^+}\frac{1}{2\epsilon}\int_0^{t_1-\epsilon}x(\tau)x^T(\tau)d\tau \nonumber\\
& +\sum_{j=1}^{p-1}\frac{1}{2\epsilon}\int_{t_j+\epsilon}^{t_{j+1}-\epsilon}x(\tau)x^T(\tau)d\tau+XX^T
\end{align}
where
\begin{align}
&\lim_{\epsilon\to 0^+}\frac{1}{2\epsilon}\int_0^{t_1-\epsilon}x(\tau)x^T(\tau)d\tau\nonumber\\
&+\sum_{j=1}^{p-1}\frac{1}{2\epsilon}\int_{t_j+\epsilon}^{t_{j+1}-\epsilon}x(\tau)x^T(\tau)d\tau\geq 0
\end{align}
and $XX^T\geq \sigma_{min}(X)I$ where $\sigma_{min}(X)>0$ as $X$ is full rank. Thus, it can be inferred that $\lim_{\epsilon\to 0^+}\frac{1}{2\epsilon}C>0$, implying $C>0$ as $\epsilon>0$. Therefore
\begin{equation}
\int_0^{t_e} x(\tau)x^T(\tau)d\tau\geq\alpha I
\end{equation} 
with $\alpha=\lambda_{min}(C)>0$.
\end{proof}



\ifCLASSOPTIONcaptionsoff
  \newpage
\fi



%


\bibliography{CL_JRNL_Ref}
\bibliographystyle{ieeetr}

%








\end{document}